\documentclass[fleqn]{article}
\usepackage{amsmath}
\usepackage{amsthm}
\usepackage{amssymb}
\newtheorem*{lem}{Lemma}
\newtheorem*{thm}{Theorem}
\newtheorem*{defn}{Definition}
\DeclareMathOperator{\Inv}{Inv}
\DeclareMathOperator{\Tr}{Tr}
\DeclareMathOperator{\ran}{ran}
\DeclareMathOperator{\spn}{span}

\begin{document}
\Large
\begin{center}
{\bf Special entangled quantum systems and the Freudenthal construction}
\end{center}
\large
\vspace*{-.1cm}
\begin{center}
P\'eter Vrana and P\'eter L\'evay  

\end{center}

\vspace*{-.4cm} \normalsize
\begin{center}
$^{1}$Department of Theoretical Physics, Institute of Physics, Budapest University of\\ Technology and Economics, H-1521 Budapest, Hungary

\vspace*{.2cm} (10 February 2009) \end{center}

\vspace*{-.3cm} \noindent \hrulefill

\vspace*{.1cm} \noindent {\bf Abstract}

\noindent We consider special quantum systems containing both distinguishable
and identical constituents. It is shown that for these systems
the Freudenthal construction based on cubic Jordan algebras naturally defines
entanglement measures invariant under the group of stochastic local operations and classical communication (SLOCC).
For this type of multipartite entanglement the SLOCC classes can be explicitly given.
These results enable further explicit constructions of multiqubit entanglement measures for distinguishable constituents by embedding them into identical fermionic ones. 
We also prove that the Pl\"ucker relations for the embedding system provide a
sufficient and necessary condition for the separability of the embedded one.
We argue that this embedding procedure
can be regarded as a convenient representation for quantum systems of
particles which are \emph{really} indistinguishable but for some reason they are not in the same state of some inner degree of freedom.

\vspace*{.3cm}
\hspace*{1.95cm} 

\vspace*{-.2cm} \noindent \hrulefill

\section{Introduction}

It has become by now a common wisdom that quantum entanglement can succesfully be regarded as a resource for performing different tasks connected with the processing of quantum information.
However, in order to fully exploit further exciting possibilities still hidden in this resource it has to be carefully quantified and its different types classified.
The desire to achieve this goal has initiated a detailed mathematical study
on the construction and structure of entanglement measures with particular physical relevance.
Such measures (or witnesses) are real-valued functions on pure or mixed entangled states
capable of grasping a particular aspect of quantum entanglement of the physical system such states represent.
For a subclass of entanglement measures the basic property is their invariance under a 
particular group of transformations  representing admissible local operations
on the entangled subsystems of the physical system.
Two important class of such transformations are the group of local unitary transformations (LU) and the so called SLOCC group\cite{Dur} 
corresponding to stochastic local operations and classical communication.
 For the LU and SLOCC classification of
 systems of {\it distinguishable} constituents characterized by
 either pure or mixed states a great variety of results is available
 \cite{Virmani,Horodecki,Bengtsson}.
 
  However, much less is known about the structure of multipartite
   entanglement measures and the corresponding entanglement classes 
   for systems with
   {\it indistinguishable} constituents. For {\it bipartite}
   fermionic and bosonic systems a number of useful results
   exists\cite{Annals,Li,Paskauskas,Ghirardi,Sanders,Schlie,Gittings,LNP}
   related to the existence of a variant of the conventional Schmidt decomposition.
For the classification of multipartite entanglement of fermionic and bosonic sytems Eckert et.al.\cite{Annals} provided an analysis with some hints on          how useful entanglement measures should be constructed.
As a next step in our recent paper we have shown that for tripartite fermionic systems with six single particle states a genuine measure of tripartite entanglement exists and the corresponding SLOCC classification can be fully given\cite{LPVP}. 
The striking feature of this classification is its similarity with the corresponding one found for three-qubits\cite{Dur}. 
We have shown that this correspondence between a tripartite system with identical and a tripartite one with indistinguishable constituents has its roots in their common underlying mathematical structure related to Freudenthal systems
based on cubic Jordan algebras\cite{Mc,Freudenthal,Krutelevich}.

Freudenthal triple systems have already made their debut to physics within the realm of string theory and supergravity. Such systems can be used as a representation
of the charge vector space occurring when studying certain black hole solutions in $N=8$, $d=4$ supergravity\cite{Gun1,Gun2,Gun3,Mal1,Ferr1}.
Recently striking multiple relations have been established between the physics of such stringy black hole solutions and quantum information theory\cite{Duff1,Linde,Lev1,Ferr2,Lev2,Lev3,Ferr3,Duff2, Duff3, Lev4}.
Though this "black hole analogy" still begs for a physical basis, the underlying correspondences have repeatedly proved to be useful for obtaining new insights into one of these fields by exploiting the methods established within the other.
Since Freudenthal triples have already turned out to be important in the string theoretic context the idea is to use this algebraic structure also in quantum information\cite{Leron1,Duff3}. A latest variation on this theme that appeared in the literature is a Freudenthal triple based reconsideration of three qubit entanglement\cite{Leron2}.
Following this trend in this paper we extend the range of applications of Freudenthal systems also to include entangled systems containing {\it both} distinguishable and indistinguishable constituents.
The basic idea of our investigation is the one of embedding one type of system
into another one.
Since Freudenthal systems are too special retaining merely the idea of embedding from them, in the second half of the paper we  explore this principle in a more general context.

The plan of the paper is as follows. In Section 2. we consider Freudenthal systems giving rise to combined systems containing bosonic and fermionic constituents.
After reconsidering the case of three fermions with six single particle states, we discuss the systems containing one distinguished qubit and two bosonic qubits, three bosonic qubits, and a one with one qubit and two fermions with four single particle states. Then we give the representatives of SLOCC equivalence classes.
As a generalization in Section 3. we discuss systems with distinguishable constituents that can be embedded into fermionic ones.
We investigate issues of separability, and their relation to the Pl\"ucker relations well-known from multilinear algebra.
Then we concentrate on  a phenomenon called the splitting of SLOCC classes a method which might prove to be a useful tool in providing entanglement classification in more complicated systems.
As a next step we introduce a family of SLOCC-invariants for fermionic systems.
We conclude with an explicit calculation of reduced density matrices
for our embedded systems.
A possible physical interpretation of our embedding procedure is briefly mentioned.

\section{Freudenthal systems}

In our recent paper\cite{LPVP} we have shown that a genuine tripartite          entanglement
measure for three fermions
with six single particle states can be constructed                              using Freudenthal's construction
applied to the cubic Jordan algebra of $J_3=M_3(\mathbb{C})$.
As a starting point in this section we recap our results.

The Jordan algebra in question is the linear space of
$3\times 3$ complex matrices equipped with the Jordan product $x\bullet y=\frac{
1}{2}(xy+yx)$.
On this Jordan algebra we have the cubic norm form $N(A)=\det A$, the sharp map
$A\mapsto A^{\sharp}$ satisfying $AA^{\sharp}=A^{\sharp}A=N(A)I_{3}$ for all    $A\in J_3$
and
the trace bilinear form $(A,B)\mapsto\Tr(AB)$.
(For a brief summary of the relevant material on such                           structures see the Appendix.)
It can be shown that the group of transformations of the                        Freudenthal triple system
$\mathfrak{M}=\mathbb{C}\oplus\mathbb{C}\oplus J_{3}\oplus J_{3}$ preserving its quartic
form is precisely $SL(6,\mathbb{C})$ and the representation of this group on
$\mathfrak{M}$ is isomorphic to $\bigwedge^{3}V_{6}$ where $V_{n}$              denotes the standard
representation of $SL(n,\mathbb{C})$ on the vector space of                     $n$-tuples of complex
numbers. This group is a subgroup of the SLOCC (stochastic local operations and classical communication\cite{Dur}) group i.e. $GL(6, \mathbb{C})$ acting for three fermions with six single particle states.

An isomorphism clarifying such issues can be explicitely given as follows. Let $\{e_{1},e_{2},e_{3},$ $e_{\bar{1}}\equiv e_{4},e_{\bar{2}}\equiv e_{5},e_{\bar{3}}\equiv e_{6}\}$
be an orthonormal basis of $\mathbb{C}^6$, and let $x\wedge y\wedge z$ denote
the normalized wedge product of the vectors $x,y,z\in\mathbb{C}^6$:
\begin{eqnarray}
x\wedge y\wedge z
& = & \frac{1}{\sqrt{6}}(x\otimes y\otimes z+y\otimes z\otimes x+z\otimes x\otimes y \nonumber \\
& & {}-x\otimes z\otimes y-z\otimes y\otimes x-y\otimes x\otimes z)
\end{eqnarray}
Using these notations a three-fermion state may be written as
\begin{equation}
P=\sum_{1\le a<b<c\le \bar{3}}P_{abc}e_{a}\wedge e_{b}\wedge e_{c}
\end{equation}
with the $20$ independent coefficients satisfying the condition
\begin{equation}
\sum_{1\le a<b<c\le \bar{3}}|P_{abc}|^{2}=1
\end{equation}
meaning that the norm of the state is $1$. The corresponding element
of $\mathfrak{M}$ is $x=(\alpha,\beta,A,B)$ where
\begin{equation}
\alpha=P_{123}\quad\beta=P_{\bar{1}\bar{2}\bar{3}}\quad
A=\left[\begin{array}{ccc}
P_{1\bar{2}\bar{3}}  &  P_{1\bar{3}\bar{1}}  &  P_{1\bar{1}\bar{2}}  \\
P_{2\bar{2}\bar{3}}  &  P_{2\bar{3}\bar{1}}  &  P_{2\bar{1}\bar{2}}  \\
P_{3\bar{2}\bar{3}}  &  P_{3\bar{3}\bar{1}}  &  P_{3\bar{1}\bar{2}}
\end{array}\right]\quad
B=\left[\begin{array}{ccc}
P_{\bar{1}23}  &  P_{\bar{1}31}  &  P_{\bar{1}12}  \\
P_{\bar{2}23}  &  P_{\bar{2}31}  &  P_{\bar{2}12}  \\
P_{\bar{3}23}  &  P_{\bar{3}31}  &  P_{\bar{3}12}
\end{array}\right]
\end{equation}
Then the quartic polynomial preserved by the action of $SL(6,\mathbb{C})$ is
\begin{equation}
T=4([\Tr(AB)-\alpha\beta]^2-4\Tr(A^{\sharp}B^{\sharp})+4\alpha\det A + 4\beta\det B )
\end{equation}
and the tripartite entanglement measure is $\mathcal{T}_{123}=|T|$. Since under
the action of $GL(6,\mathbb{C})$ this quantity takes up a nonzero factor, one
immediately concludes that there must be at least two SLOCC equivalence classes
of three-fermion states. In fact, by introducing the dual state and utilizing
Pl\"ucker's relations one can complete the classification and it turns out that
we have four SLOCC orbits: the separable one, the biseparable one, and two different
types of true tripartite entanglement\cite{LPVP}.

\subsection{Three qubits}

This classification resembles the one of the three-qubit system\cite{Dur,Leron1}, where the well-known
representatives of the inequivalent classes with tripartite entanglement are
the $W$ and $GHZ$ states. One may suspect that there is a connection between the
two, and this indeed is the case. By looking at special three-fermion states one
may observe that the space of three-qubit states $\bigotimes_{i=1}^{3}\mathbb{C}^2$
can be injected into our three-fermion one in such a way that the three-tangle\cite{Kundu}
defined by Cayley's hyperdeterminant can be viewed as a special case of the
quartic above.

To this end we keep only the amplitudes with three different numbers in the subscript
forgetting the overbars for a moment. We have $8$ such coefficients which is the
number of coefficients needed to describe a three-qubit state. The natural way
to do this is to choose an orthonormal basis $\{f_{0},f_{1}\}\in \mathbb{C}^{2}$ and take
three-fold tensor products of them (computational basis). Now let us map an
element $f_{i}\otimes f_{j}\otimes f_{k}$ of this basis to
$e_{1+3i}\wedge e_{2+3j}\wedge e_{3+3k}\in\bigwedge^{3}\mathbb{C}^{6}$ i.e. the
overbar indicates $1$ and the lack of it indicates $0$ on the place indexed by
the number in the subscript. To a three-qubit state
\begin{equation}
a=\sum_{i,j,k\in\{0,1\}}a_{ijk}f_{i}\otimes f_{j}\otimes f_{k}
\end{equation}
we associate this way an element $x=(\alpha,\beta,A,B)$ of $\mathfrak{M}$ where
\begin{equation}
\alpha=a_{000}\quad\beta=a_{111}\quad
A=\left[\begin{array}{ccc}
a_{011}  &  0  &  0  \\
0  &  a_{101}  &  0  \\
0  &  0  &  a_{110}
\end{array}\right]\quad
B=\left[\begin{array}{ccc}
a_{100}  &  0  &  0  \\
0  &  a_{010}  &  0  \\
0  &  0  &  a_{001}
\end{array}\right]
\end{equation}
For this element $\mathcal{T}_{123}$ equals the three-tangle of $a$ (using decimal
notation):
\begin{eqnarray}
T & = & 4((a_0a_7)^2+(a_1a_6)^2+(a_2a_5)^2+(a_3a_4)^2)  \nonumber \\
  & & -8(a_0a_7a_1a_6+a_0a_7a_2a_5+a_0a_7a_3a_4  \nonumber \\
  & & {}+a_1a_6a_2a_5+a_1a_6a_3a_4+a_2a_5a_3a_4)  \nonumber \\
  & & +16(a_0a_3a_5a_6+a_7a_4a_2a_1)
\end{eqnarray}
In this way the three-qubit system can be embedded in the three-fermion one.

An other way to look at the similarity between the two systems is obtained by
observing that starting with the cubic Jordan algebra $J_{1+1+1}=\mathbb{C}\oplus\mathbb{C}\oplus\mathbb{C}$
the Freudenthal construction leads to the $V_{2}\otimes V_{2}\otimes V_{2}$ representation
of the group $SL(2,\mathbb{C})^3$, and the quartic polynomial preserved by the
action of the group is Cayley's hyperdeterminant. In the appendix it is shown
that $J_{1+1+1}$ is isomorphic to the subalgebra of $J_{3}$ of diagonal matrices.

For an element $x=(x_{1},x_{2},x_{3})\in J_{1+1+1}$ we have a cubic norm $N(x)=x_{1}x_{2}x_{3}$,
the sharp map assigning $x^{\sharp}=(x_{2}x_{3},x_{1}x_{3},x_{1}x_{2})$ to $x$ and
on $J$ we have a bilinear form whose value is $(x,y)=x_{1}y_{1}+x_{2}y_{2}+x_{3}y_{3}$
for $y=(y_{1},y_{2},y_{3})$.

We see that the injection of the space of three-qubit states into the space of
three-fermion states can be done at the Jordan algebra level. Moreover the quartic
invariant is based entirely on the cubic Jordan algebra structure in both cases.
It is not surprising therefore that the three-tangle of a three-qubit state can
be obtained also by first taking the associated special three-fermion state and
then calculating the value of the quartic invariant on it.

\subsection{One distinguished qubit with two bosonic qubits}

According to the literature on Freudenthal triple systems\cite{Krutelevich} there are three more Jordan
algebras for which the Freudenthal construction yields a representation that
has a natural interpretation in quantum information theory. These are
$J_{1}=\mathbb{C}$, $J_{1+1}=\mathbb{C}\oplus\mathbb{C}$ and
$J_{1+2}=\mathbb{C}\oplus M_{2}(\mathbb{C})$. It is shown in the Appendix
that these are all isomorphic to subalgebras of $J_{3}$.

The first two cases correspond to three bosonic qubits and a composite system
consisting of one qubit and two other indistinguishable bosonic qubits. The
Hilbert space of these systems can be naturally viewed as subspaces of the one
describing three qubits so one might expect that these can be embedded into the
latter much like the three-qubit system is embedded in the three-fermion one.
The last one corresponds to a system containing a qubit and two indistinguishable
fermionic particles with four single particle states.

First take a look at $J_{1+1}$. With the Freudenthal construction we obtain a
representation of $SL(2,\mathbb{C})^2$ on
$\mathbb{C}\oplus\mathbb{C}\oplus J_{1+1}\oplus J_{1+1}$ that is isomorphic to
$V_{2}\otimes Sym^{2}V_{2}$. This enables us to classify entangled states in
the space of a distinguishable and two bosonic qubits.

Let $\{e_{0},e_{1}\}$ be the computational basis of $\mathbb{C}^2$, and let
$f_{0}=e_{0}\otimes e_{0}$, $f_{1}=e_{0}\otimes e_{1}+e_{1}\otimes e_{0}$ and $f_{2}=e_{1}\otimes e_{1}$.
Now a normalized vector in $\mathbb{C}^2\otimes Sym^{2}\mathbb{C}^2$ may be
written as
\begin{equation}
b=\sum_{i=0}^{1}\sum_{j=0}^{2}b_{ij}e_{i}\otimes f_{j}
\end{equation}
where
\begin{equation}
\sum_{i=0}^{1}(|b_{i0}|^2+2|b_{i1}|^2+|b_{i2}|^2)=1
\end{equation}
The corresponding three-fermion state is given by $x=(\alpha,\beta,A,B)\in\mathfrak{M}$ where
\begin{equation}
\alpha=b_{00}\quad\beta=b_{12}\quad
A=\left[\begin{array}{ccc}
b_{02}  &  0  &  0  \\
0  &  b_{11}  &  0  \\
0  &  0  &  b_{11}
\end{array}\right]\quad
B=\left[\begin{array}{ccc}
b_{10}  &  0  &  0  \\
0  &  b_{01}  &  0  \\
0  &  0  &  b_{01}
\end{array}\right]
\end{equation}
For the state $b$ we have the quartic invariant:
\begin{eqnarray}
T & = & 4(b_{00}^{2}b_{12}^{2}+b_{02}^{2}b_{10}^{2})+16(b_{11}^{2}b_{00}b_{02}+b_{01}^2b_{10}b_{12})  \nonumber \\
  & & -8b_{00}b_{02}b_{10}b_{12}-16(b_{01}b_{02}b_{10}b_{11}+b_{00}b_{01}b_{11}b_{12})
\end{eqnarray}
We are not aware of any application of the quartic invariant above within quantum information theory. However, it is interesting to note that our invariant appears within the realm of  black hole solutions in string theory and quantum gravity.
The model in question is the so called $st^2$ model\cite{Bellucci1} which can be regarded as
a one coming from the stu-model\cite{Duff4,Behrndt} after a $t=u$ degeneracy.
In this model the black hole entropy is expressed in terms of $6$ charges ($3$ magnetic and $3$ electric) which can be mapped bijectively to the $6$ amplitudes of our state $b$.
This correspondence between an entanglement measure on one side and the black hole entropy formula for a particular black hole solution on the other forms the basis of the black hole analogy our guiding principle in constructing this measure.

\subsection{Three bosonic qubits}

Now let us turn to the Jordan algebra $J_{1}=\mathbb{C}$ in which the norm of
an element is simply the cube of it, the sharp means taking the square, and the
trace bilinear form of two elements $x,y\in J_{1}$ is $3xy$. Again after some
calculation one can show that $J_{1}$ is a subalgebra of $J_{1+1}$ the inclusion
map being $x\mapsto(x,x)$. The Freudenthal construction in this case leads to
a four dimensional representation of $SL(2,\mathbb{C})$ isomorphic to $Sym^3 V_{2}$
which is related to the system of three indistinguishable bosonic qubits.

A general normalized state in $Sym^3 \mathbb{C}^2$ may be written as
\begin{eqnarray}
c & = & c_{0}e_{0}\otimes e_{0}\otimes e_{0}+c_{3}e_{1}\otimes e_{1}\otimes e_{1} \nonumber \\
  & & +c_{1}(e_{1}\otimes e_{0}\otimes e_{0}+e_{0}\otimes e_{1}\otimes e_{0}+e_{0}\otimes e_{0}\otimes e_{1}) \nonumber \\
  & & +c_{2}(e_{0}\otimes e_{1}\otimes e_{1}+e_{1}\otimes e_{0}\otimes e_{1}+e_{1}\otimes e_{1}\otimes e_{0})
\end{eqnarray}
where $|c_{0}|^2+|c_{3}|^2+3(|c_{1}|^2+|c_{2}|^2)=1$. To this we associate the
element $x=(c_{0},c_{3},c_{2}I_{3},c_{1}I_{3})$ in $\mathfrak{M}$ where $I_n$ denotes the
identity element of $M_{n}(\mathbb{C})$. For this state the quartic invariant
is:
\begin{equation}
T=4c_{0}^{2}c_{3}^{2}-12c_{1}^{2}c_{2}^{2}-24c_{0}c_{1}c_{2}c_{3}+16(c_{0}c_{2}^{3}+c_{3}c_{1}^{3})
\end{equation}
Our entanglement measure for three bosonic qubits has appeared in the context of stringy black holes as the black hole entropy formula in the so called $t^3$-model\cite{Bellucci1,Vafa}.
For the interesting geometry of three bosonic qubits see the paper of Brody et.al.\cite{Brody}.

\subsection{One qubit and two fermions with four single particle states}

The remaining Jordan algebra is $J_{1+2}=\mathbb{C}\oplus M_{2}(\mathbb{C})$.
Let $x=(\alpha,x_{0})$ and $y=(\beta,y_{0})$ be two elements of $J_{1+2}$.
The cubic norm form is given by $N(x)=\alpha\det x_{0}$, the sharp map is
$x^{\sharp}=(\det x_{0},\alpha(\Tr x_{0})I_{2}-\alpha x_{0})$, and finally the
trace bilinear map in this case is $(x,y)\mapsto \alpha\beta+\Tr(x_0y_0)$.
Once again one can check that this Jordan algebra can naturally be viewed as a
subalgebra of $J_{3}$ namely it is isomorphic to the subalgebra of block-diagonal
matrices with a $1\times 1$ and a $2\times 2$ block in the diagonal.

We know that the Freudenthal construction applied to $J_{1+2}$ yields a
representation of $SL(2,\mathbb{C})\times SL(4,\mathbb{C})$ isomorphic to
$V_{2}\otimes\bigwedge^{2}V_{4}$. Moreover it is known that the fermionic
measure describing the bipartite entanglement of two fermions with four single
particle states reduces to the two-qubit concurrence\cite{Gittings,LPVP} in the same way as the
three-fermion measure reduces to the three-tangle, hence one may expect that in
a sense this system fits between the three-fermion and the three-qubit one. This
expectation is further supported by the fact that we have the embeddings
$J_{1+1+1}\subset J_{1+2}\subset J_{3}$.

Let us see how this works explicitely. Let $\{e_{0},e_{1}\}$ and $\{f_{0},f_{1},f_{2},f_{3}\}$
be the canonical basis of $\mathbb{C}^{2}$ and $\mathbb{C}^{4}$ respectively.
A state $d\in\mathbb{C}^2\otimes\bigwedge^{2}\mathbb{C}^4$ may be written as
\begin{equation}
d=\sum_{i=0}^{1}\sum_{0\le j<k\le 3}d_{ijk}e_{i}\otimes(f_{j}\wedge f_{k})
\end{equation}
the amplitudes being antisymmetric in the second and third index. The condition
of being normalized means that
\begin{equation}
\sum_{i=0}^{1}\sum_{0\le j<k\le 3}|d_{ijk}|^{2}=1
\end{equation}
Now we relate the six-state labels to the two-state and four-state ones as
$(1,\bar{1})\mapsto(0,1)$ and $(2,3,\bar{2},\bar{3})\mapsto(0,1,2,3)$ respectively, and
keep only the 12 coefficients whose index contains precisely one of $1$ and $\bar{1}$.
We associate to $d$ the element $x=(\alpha,\beta,A,B)\in\mathfrak{M}$ where
\begin{equation}
\alpha=d_{001}\quad\beta=d_{123}\quad
A=\left[\begin{array}{ccc}
d_{023}  &  0  &  0 \\
0  &  d_{103}  &  d_{120} \\
0  &  d_{113}  &  d_{121}
\end{array}\right]\quad
B=\left[\begin{array}{ccc}
d_{101}  &  0  &  0 \\
0  &  d_{021}  &  d_{002} \\
0  &  d_{031}  &  d_{003}
\end{array}\right]
\end{equation}
For this state the value of the quartic tripartite entanglement measure is
\begin{eqnarray}
T
 & = & 4((d_{023}d_{101})^2+(d_{021}d_{103})^2+(d_{002}d_{113})^2  \nonumber \\
 & & +(d_{031}d_{120})^2+(d_{003}d_{121})^2+(d_{001}d_{123})^2)  \nonumber \\
 & & +8(d_{002}d_{021}d_{103}d_{113}+d_{021}d_{031}d_{103}d_{120}  \nonumber \\
 & & {}+d_{002}d_{003}d_{113}d_{121}+d_{003}d_{031}d_{120}d_{121})  \nonumber \\
 & & +16(d_{003}d_{021}d_{113}d_{120}+d_{001}d_{023}d_{103}d_{121}  \nonumber \\
 & & {}+d_{002}d_{031}d_{103}d_{121}+d_{003}d_{021}d_{101}d_{123})  \nonumber \\
 & & -16(d_{001}d_{023}d_{113}d_{120}+d_{002}d_{031}d_{101}d_{123})  \nonumber \\
 & & -8(d_{021}d_{023}d_{101}d_{103}+d_{002}d_{023}d_{101}d_{113}+d_{023}d_{031}d_{101}d_{120}  \nonumber \\
 & & {}+d_{002}d_{031}d_{113}d_{120}+d_{003}d_{023}d_{101}d_{121}+d_{003}d_{021}d_{103}d_{121}  \nonumber \\
 & & {}+d_{001}d_{023}d_{101}d_{123}+d_{001}d_{021}d_{103}d_{123}+d_{001}d_{002}d_{113}d_{123}  \nonumber \\
 & & {}+d_{001}d_{031}d_{120}d_{123}+d_{001}d_{003}d_{121}d_{123})
\end{eqnarray}

\subsection{Representatives of SLOCC equivalence classes}

\begin{table}[b]
\centering
\begin{tabular}{c|c|l}
$\mathcal{H}$  &  $G$  &  remark  \\
\hline
$\mathcal{H}_0=\bigwedge^3\mathbb{C}^6$  &  $G_0=GL(6,\mathbb{C})$&  \\
$\mathbb{C}^2\otimes\bigwedge^2\mathbb{C}^4$  &  $GL(2,\mathbb{C})\times GL(4,\mathbb{C})$  &  $A,B\in M_{1}(\mathbb{C})\oplus M_{2}(\mathbb{C})$\\
$\mathbb{C}^2\otimes\mathbb{C}^2\otimes\mathbb{C}^2$  &  $GL(2,\mathbb{C})^3$  &  $A,B$ diagonal  \\
$\mathbb{C}^2\otimes Sym^2\mathbb{C}^2$  &  $GL(2,\mathbb{C})^2$  &  $A,B\in M_{1}(\mathbb{C})\oplus \mathbb{C}\cdot I_{2}$  \\
$Sym^3\mathbb{C}^2$  &  $GL(2,\mathbb{C})$  &  $A,B\in\mathbb{C}\cdot I_{3}$  \\

\end{tabular}
\caption{\label{tab:subs}Subspaces of $\mathfrak{M}$ associated to Hilbert spaces $\mathcal{H}\subset\mathcal{H}_0$ describing various quantum mechanical systems with SLOCC group $G\subset G_0$.}
\end{table}
To sum up, we have the chain of embeddings of Jordan algebras $J_{1}\subset J_{1+1}\subset J_{1+1+1}\subset J_{1+2}\subset J_{3}$
that gives rise via Freudenthal's construction to the chain of embeddings of
Hilbert spaces $Sym^{3}\mathbb{C}^2\subset \mathbb{C}^2\otimes Sym^2\mathbb{C}^2\subset\mathbb{C}^2\otimes\mathbb{C}^2\otimes\mathbb{C}^2\subset\mathbb{C}^2\otimes\bigwedge^2\mathbb{C}^4\subset\bigwedge^{2}\mathbb{C}^6$.
The appropriate subspaces of $\mathfrak{M}$ are shown in Table \ref{tab:subs}
along with their SLOCC group. These embeddings are compatible with the SLOCC
classification of entanglement in the sense that SLOCC orbits of any of these
systems are subsets of the intersections of SLOCC orbits of the three-fermion
Hilbert space with the appropriate subspace. In order to find representatives
of various entanglement classes it is therefore enough to look for them in the
smallest possible subspace then interpret it as elements of the larger Hilbert
spaces. These representatives can be chosen to be the following ones:
\begin{eqnarray}
GHZ & = & \frac{1}{\sqrt{2}}(1,1,0,0)  \\
W & = & \frac{1}{\sqrt{3}}(0,0,0,I_{3})  \\
B_{1} & = & \frac{1}{\sqrt{2}}(1,0,\left[\begin{array}{ccc}
1 & 0 & 0 \\
0 & 0 & 0 \\
0 & 0 & 0
\end{array}\right],0)  \\
B_{2} & = & \frac{1}{\sqrt{2}}(1,0,\left[\begin{array}{ccc}
0 & 0 & 0 \\
0 & 1 & 0 \\
0 & 0 & 0
\end{array}\right],0)  \\
B_{3} & = & \frac{1}{\sqrt{2}}(1,0,\left[\begin{array}{ccc}
0 & 0 & 0 \\
0 & 0 & 0 \\
0 & 0 & 1
\end{array}\right],0)  \\
S & = & (1,0,0,0)
\end{eqnarray}
The $GHZ$ and $W$ states show tripartite entanglement, $B_{i}$ is a biseparable and
$S$ is a separable state. Apart from $B_{i}$ these can be found in the system of
three bosonic qubits, but the relations characterizing states of rank at most two
imply separability in this case. Therefore the representative of the biseparable
class is chosen from the larger Hilbert space $\mathbb{C}^2\otimes Sym^{2}\mathbb{C}^2$.
Of the biseparable subclasses only $B_{1}$ is present in the latter, all can be
found in the three-qubit case, $B_{2}$ and $B_{3}$ are equivalent in
$\mathbb{C}^{2}\otimes\bigwedge^2\mathbb{C}^4$, and all three are equivalent in
the largest Hilbert-space $\bigwedge^3\mathbb{C}^6$.
Table \ref{tab:reps} shows these states for each system.

\begin{table}[h]
\centering
\begin{tabular}{c|c}
space ($\mathcal{H}$)  &  representatives  \\
\hline
  &  $GHZ=\frac{1}{\sqrt{2}}(e_{1}\wedge e_{2}\wedge e_{3}+e_{4}\wedge e_{5}\wedge e_{6})$  \\
$\bigwedge^3\mathbb{C}^6$
  &  $W=\frac{1}{\sqrt{3}}(e_{4}\wedge e_{2}\wedge e_{3}+e_{1}\wedge e_{5}\wedge e_{3}+ e_{1}\wedge e_{2}\wedge e_{6})$  \\
  &  $B_{1}=\frac{1}{\sqrt{2}}(e_{1}\wedge e_{2}\wedge e_{3}+e_{1}\wedge e_{5}\wedge e_{6})$  \\
  &  $S=e_{1}\wedge e_{2}\wedge e_{3}$  \\
\hline
  &  $GHZ=\frac{1}{\sqrt{2}}(e_{0}\otimes(f_{0}\wedge f_{1})+e_{1}\otimes(f_{2}\wedge f_{3}))$  \\
$\mathbb{C}^2\otimes\bigwedge^2\mathbb{C}^4$
  &  $W=\frac{1}{\sqrt{3}}(e_{0}\otimes(f_{2}\wedge f_{3})+e_{1}\otimes(f_{0}\wedge f_{3})+e_{1}\otimes(f_{2}\wedge f_{1}))$  \\
  &  $B_{1}=\frac{1}{\sqrt{2}}e_{0}\otimes(f_{0}\wedge f_{1}+f_{2}\wedge f_{3})$  \\
  &  $B_{2}=\frac{1}{\sqrt{2}}(e_{0}\otimes(f_{0}\wedge f_{1})+e_{1}\otimes(f_{0}\wedge f_{3}))$  \\
  &  $S=e_{0}\otimes(f_{0}\wedge f_{1})$  \\
\hline
  &  $GHZ=\frac{1}{\sqrt{2}}(e_{0}\otimes e_{0}\otimes e_{0}+e_{1}\otimes e_{1}\otimes e_{1})$  \\
$\mathbb{C}^2\otimes\mathbb{C}^2\otimes\mathbb{C}^2$
  &  $W=\frac{1}{\sqrt{3}}(e_{1}\otimes e_{0}\otimes e_{0}+e_{0}\otimes e_{1}\otimes e_{0}+e_{0}\otimes e_{0}\otimes e_{1})$  \\
  &  $B_{1}=\frac{1}{\sqrt{2}}(e_{0}\otimes (e_{0}\otimes e_{0}+e_{1}\otimes e_{1}))$  \\
  &  $B_{2}=\frac{1}{\sqrt{2}}(e_{0}\otimes e_{0}\otimes e_{0}+e_{1}\otimes e_{0}\otimes e_{1})$  \\
  &  $B_{3}=\frac{1}{\sqrt{2}}((e_{0}\otimes e_{0}+e_{1}\otimes e_{1})\otimes e_{0})$  \\
  &  $S=e_{0}\otimes e_{0}\otimes e_{0}$  \\
\hline
  &  $GHZ=\frac{1}{\sqrt{2}}(e_{0}\otimes(e_{0}\otimes e_{0})+e_{1}\otimes(e_{1}\otimes e_{1}))$  \\
$\mathbb{C}^2\otimes Sym^2\mathbb{C}^2$
  &  $W=\frac{1}{\sqrt{3}}(e_{1}\otimes(e_{0}\otimes e_{0})+e_{0}\otimes(e_{1}\otimes e_{0}+e_{0}\otimes e_{1}))$  \\
  &  $B_{1}=\frac{1}{\sqrt{2}}(e_{0}\otimes (e_{0}\otimes e_{0}+e_{1}\otimes e_{1}))$  \\
  &  $S=e_{0}\otimes(e_{0}\otimes e_{0})$  \\
\hline
  &  $GHZ=\frac{1}{\sqrt{2}}(e_{0}\otimes e_{0}\otimes e_{0}+e_{1}\otimes e_{1}\otimes e_{1})$  \\
$Sym^3\mathbb{C}^2$
  &  $W=\frac{1}{\sqrt{3}}(e_{1}\otimes e_{0}\otimes e_{0}+e_{0}\otimes e_{1}\otimes e_{0}+e_{0}\otimes e_{0}\otimes e_{1})$  \\
  &  $S=e_{0}\otimes e_{0}\otimes e_{0}$  \\
\end{tabular}
\caption{\label{tab:reps}Representatives of SLOCC orbits of quantum mechanical systems classified via Freudenthal's construction.}
\end{table}

\section{Embedding systems with distinguishable constituents into fermionic ones}

Looking at these results one might ask what parts of this process can be
generalized. Since the systems considered above are the only ones which are
related to Freudenthal triples we can not extend these results for more
general quantum systems. On the other hand the phenomenon that the Hilbert
space of a smaller system is embedded in the space of a larger one appears
in more general cases.

To be more specific, suppose that we have the invariants $I_{1},\ldots,I_{n}$
under the SLOCC group $G$ represented on a Hilbert space $\mathcal{H}$. It may
happen that there exists a subgroup $G'\subset G$ that can be viewed as a SLOCC
group of a quantum mechanical system to which we can associate a subspace
$\mathcal{H}'\subset\mathcal{H}$ of our original Hilbert space in such a way
that $\mathcal{H}'$ is invariant under the action of $G'$. Obviously the
invariants $I_{1},\ldots,I_{n}$ restricted to $\mathcal{H}'$ are invariant
under the subgroup $G'$ of SLOCC transformations of the smaller system. This
means that if we have two states in $\mathcal{H}'$ that are $G$-inequivalent,
then they will necessarily be $G'$-inequivalent. However, it may also happen
that two $G$-equivalent states cannot be transformed to each other by the smaller
SLOCC group $G'$, in other words, the intersection of a SLOCC orbit in $\mathcal{H}$
with $\mathcal{H}'$ regarded as a set in $\mathcal{H}'$ may split to several
disjoint orbits under the smaller SLOCC group $G'$. In this case further
refinement is needed to obtain full classification of the entangled states in
$\mathcal{H}'$.

\subsection{Separability in arbitrary systems}

As an example suppose that we have $N$ types of fermionic particles, $k_{i}$
of the $i^{\mathrm{th}}$ type having $n_{i}$ single particle states. To this composite
system we associate the Hilbert space
$\mathcal{H}=\bigwedge^{k_{1}}\mathcal{H}_{1}^{(0)}\otimes\cdots\otimes\bigwedge^{k_{N}}\mathcal{H}_{N}^{(0)}$
where $\dim\mathcal{H}_{i}^{(0)}=n_{i}$ and the SLOCC group $G=GL(n_{1},\mathbb{C})\times\cdots\times GL(n_{N},\mathbb{C})$.
This space can be embedded in $\mathcal{K}=\bigwedge^{k_{1}+\cdots+k_{N}}(\mathcal{H}_{1}^{(0)}\oplus\cdots\oplus\mathcal{H}_{N}^{(0)})$
via the linear map $\varphi:\mathcal{H}\to\mathcal{K}$ defined by
\begin{multline}
\varphi:(v_{1}\wedge\ldots\wedge v_{k_{1}})\otimes(v_{k_{1}+1}\wedge\ldots\wedge v_{k_{1}+k_{2}})\otimes\cdots\otimes(v_{k_{1}+\cdots+k_{N-1}+1}\wedge\ldots\wedge v_{k_{1}+\cdots+k_{N}}) \\ \mapsto v_{1}\wedge v_{2}\wedge\ldots\wedge v_{k_{1}+\cdots+k_{N}}
\end{multline}
on tensor products of decomposable vectors and its SLOCC group
$H=GL(n_{1}+\cdots+n_{N},\mathbb{C})$ contains $G$ in an obvious way. Before
continuing we need some preliminary linear algebraic facts stated in what follows.

\begin{lem}
Let $\mathcal{H}_{A}\simeq\mathbb{C}^{n_{A}}$ and $\mathcal{H}_{B}\simeq\mathbb{C}^{n_{B}}$
be two Hilbert-spaces and $\mathcal{H}=\bigwedge^{k_{A}+k_{B}}(\mathcal{H}_{A}\oplus\mathcal{H}_{B})$
for some $k_{A},k_{B}\in\mathbb{N}$. Suppose that $v\in\mathcal{H}\setminus\{0\}$
is decomposable and $v\in\spn\{a_{i_{1}}\wedge\ldots\wedge a_{i_{k_{A}}}\wedge b_{i_{k_{A}+1}}\wedge\ldots\wedge b_{i_{k_{A}+k_{B}}}\vert 1\le i_{1},\ldots,i_{k_{A}}\le n_{A},1\le i_{k_{A}+1},\ldots,i_{k_{A}+k_{B}}\le n_{B}\}=\mathcal{H}_{0}$
where $\{a_{i}\}_{i=1}^{n_{A}}$ and $\{b_{i}\}_{i=1}^{n_{B}}$ are orthonormal
bases in $\mathcal{H}_{A}$ and $\mathcal{H}_{B}$ respectively.
Then $v$ can also be written in the form $v^{A}_{1}\wedge\ldots\wedge v^{A}_{k_{A}}\wedge v^{B}_{1}\wedge\ldots\wedge v^{B}_{k_{B}}$
where $\{v^{A}_{i}\}_{i=1}^{k_{A}}\subset\mathcal{H}_{A}$ and $\{v^{B}_{i}\}_{i=1}^{k_{B}}\subset\mathcal{H}_{B}$.
\end{lem}
\begin{proof}
Let $k=k_{A}+k_{B}$ and $v=v_{1}\wedge v_{2}\wedge\ldots\wedge v_{k}$ be a decomposition
of $v$. Each $v_{i}$ can be uniquely written in the form $v_{i}=A_{i}+B_{i}$
where $A_{i}\in\mathcal{H}_{A}$ and $B_{i}\in\mathcal{H}_{B}$. Now take a look
at the terms in the expanded form of $v=(A_{1}+B_{1})\wedge\ldots\wedge(A_{k}+B_{k})$.
A simple observarion is that two wedge products of elements of $\mathcal{H}_{A}$
and $\mathcal{H}_{B}$ are orthogonal with respect to the induced inner product
if the number of factors from $\mathcal{H}_{A}$ (or $\mathcal{H}_{B}$) differ
in the two products. This means that in $v$ the term $B_{1}\wedge\ldots\wedge B_{k}$
is orthogonal to all the other terms. But since it is orthogonal to $\mathcal{H}_{0}$
too it must be $0$ which is equivalent to stating that the vectors $(B_{i})_{i=1}^{k}$
are linearly dependent. After some rearranging (and possibly including a minus sign)
we can assume that $B_{1}$ can be expressed in the form $\lambda_{2}B_{2}+\cdots+\lambda_{k}B_{k}$.
Now using multilinearity and that wedge product of linearly dependent vectors
is the null vector we can write
\begin{eqnarray}
v & = & (A_{1}+B_{1})\wedge\ldots\wedge(A_{k}+B_{k})  \nonumber \\
  & = & \left(A_{1}+B_{1}-\sum_{i=2}^{k}\lambda_{i}(A_{i}+B_{i})\right)\wedge(A_{2}+B_{2})\wedge\ldots\wedge(A_{k}+B_{k})  \nonumber \\
  & = & A_{1}'\wedge(A_{2}+B_{2})\wedge\ldots\wedge(A_{k}+B_{k})
\end{eqnarray}
for some $A_{1}'\in\mathcal{H}_{A}$. Similar reasoning with the term
$A_{1}'\wedge B_{2}\wedge\ldots\wedge B_{k}$ shows that we can assume that
$B_{2}$ can be written as a linear combination of $B_{3},\ldots,B_{k}$ and
so on finally arriving at the form
$v=A_{1}'\wedge\ldots\wedge A_{k_{A}}'\wedge(A_{k_{A}+1}+B_{k_{A}+1})\wedge\ldots\wedge(A_{k}+B_{k})$.
The number of $B$-s can not be further reduced since the term containing $k_{B}$
factors from $\mathcal{H}_{B}$ is not orthogonal to $\mathcal{H}_{0}$.

After expanding we see that $A_{1}'\wedge\ldots\wedge A_{k_{A}}'\wedge A_{k_{A}+1}\wedge\ldots\wedge A_{k}$
is orthogonal to the other addends and to $\mathcal{H}_{0}$ therefore the factors
are linearly dependent. By rearranging we can identify two cases: either
$A_{1}'$ or $A_{k_{A}+1}$ can be expressed as a linear combination of the
remaining factors. In the first case $A_{1}'$ can not be an element of
$\spn\{A_{2}',\ldots,A_{k_{A}}'\}$ since $v\neq 0$ and therefore we can find
an element in $\{A_{i}\}_{i=k_{A}+1}^{k}$ whose coefficient in the linear
expansion is nonzero leading us to the latter case. We can assume that
$A_{k_{A}+1}=\sum_{i=1}^{k_{A}}\mu_{i}A_{i}'+\sum_{i=k_{A}+2}^{k}\mu_{i}A_{i}$
and using this expansion one can write
\begin{eqnarray}
v & = & A_{1}'\wedge\ldots\wedge A_{k_{A}}'\wedge(A_{k_{A}+1}+B_{k_{A}+1})\wedge\ldots\wedge(A_{k}+B_{k})  \nonumber \\
  & = & A_{1}'\wedge\ldots\wedge A_{k_{A}}'\wedge\left(A_{k_{A}+1}+B_{k_{A}+1}-\sum_{i=1}^{k_{A}}\mu_{i}A_{i}'-\sum_{i=k_{A}+2}^{k}\mu_{i}(A_{i}+B_{i})\right)\wedge  \nonumber \\
  & & {}\wedge(A_{k_{A}+2}+B_{k_{A}+2})\wedge\ldots\wedge(A_{k}+B_{k})  \nonumber \\
  & = & A_{1}'\wedge\ldots\wedge A_{k_{A}}'\wedge B_{k_{A}+1}'\wedge(A_{k_{A}+2}+B_{k_{A}+2})\wedge\ldots\wedge(A_{k}+B_{k})
\end{eqnarray}
for some $B_{k_{A}+1}'\in\mathcal{H}_{B}$. Proceeding the same way we finally
arrive at the form $v=A_{1}'\wedge\ldots\wedge A_{k_{A}}'\wedge B_{k_{A}+1}'\wedge\ldots\wedge B_{k}'$
\end{proof}

\begin{thm}
Let $N\in\mathbb{N}$, $(k_{i})_{i=1}^N$ and $(n_{i})_{i=1}^N$ be $N$-tuples of
positive integers, $\mathcal{H}_{i}^{(0)}$ a Hilbert-space of dimension $n_{i}$,
$\mathcal{H}_{i}=\bigwedge^{k_{i}}\mathcal{H}_{i}^{(0)}$ (for all $i\in\{1,\ldots,N\}$),
$\mathcal{K}^{(0)}=\bigoplus_{i=1}^{N}\mathcal{H}_{i}^{(0)}$ and
$\mathcal{K}=\bigwedge^{k}\mathcal{K}^{(0)}$ where $k=\sum_{i=1}^{N}k_{i}$.
Let $\varphi:\mathcal{H}_{1}\otimes\cdots\otimes\mathcal{H}_{N}\to\mathcal{K}$
be the linear map defined by $(e_{1,j_{1}}\wedge\ldots\wedge e_{1,j_{k_{1}}})\otimes\cdots\otimes(e_{N,j_{k-k_{N}+1}}\wedge\ldots\wedge e_{N,j_{k}})\mapsto e_{1,j_{1}}\wedge\ldots\wedge e_{1,j_{k_{1}}}\wedge\ldots\wedge e_{N,j_{k-k_{N}+1}}\wedge\ldots\wedge e_{N,j_{k}}$
for some orthonormal bases $\{e_{i,j}\}_{j=1}^{k_{i}}\subset\mathcal{H}_{i}^{(0)}$.
Then a vector $v$ in $\mathcal{H}_{1}\otimes\cdots\otimes\mathcal{H}_{N}=\mathcal{H}$
is a tensor product of decomposable vectors in $\mathcal{H}_{i}$ iff $\varphi(v)\in\mathcal{K}$ is
decomposable.
\end{thm}
\begin{proof}
Since both sides of the equation defining $\varphi$ is linear in all
$e_{i,j_{l}}$-s, the image of $v=(v_{1,j_{1}}\wedge\ldots\wedge v_{1,j_{k_{1}}})\otimes\cdots\otimes(v_{N,j_{k-k_{N}+1}}\wedge\ldots\wedge v_{N,j_{k}})$
is $v_{1,j_{1}}\wedge\ldots\wedge v_{1,j_{k_{1}}}\wedge\ldots\wedge v_{N,j_{k-k_{N}+1}}\wedge\ldots\wedge v_{N,j_{k}}$
which is a decomposable element of $\mathcal{K}$.

For the converse observe that by introducing the Hilbert-spaces
$\mathcal{K}_{i}^{(0)}=\mathcal{H}_{1}^{(0)}\oplus\cdots\oplus\mathcal{H}_{i}^{(0)}$
and $\mathcal{K}_{i}=\bigwedge^{k_{1}+\cdots+k_{i}}\mathcal{K}_{i}^{(0)}$, for the linear injections
$\varphi_{i}:\mathcal{K}_{i}\otimes\mathcal{H}_{i+1}\to\mathcal{K}_{i+1};(x_{1}\wedge\ldots\wedge x_{k_{1}+\cdots+k_{i}})\otimes(y_{1}\wedge\ldots\wedge y_{k_{i+1}})\mapsto x_{1}\wedge\ldots\wedge x_{k_{1}+\cdots+k_{i}}\wedge y_{1}\wedge\ldots\wedge y_{k_{i+1}}$
we have $\varphi=\varphi_{N-1}\circ(\varphi_{N-2}\otimes id_{\mathcal{H}_{N}})\circ\ldots\circ(\varphi_{1}\otimes id_{\mathcal{H}_{3}}\otimes\ldots\otimes id_{\mathcal{H}_{N}})$.
Now take a vector $v$ in $\ran\varphi\subset\mathcal{K}$ that is decomposable.
Using the lemma above with $\mathcal{H}_{A}^{(0)}=\mathcal{K}_{N-1}^{(0)}$,
$\mathcal{H}_{B}^{(0)}=\mathcal{H}_{N}$, $k_{A}=k-k_{N}$ and $k_{B}=k_{N}$
we see that $v=x_{1}\wedge\ldots\wedge x_{k-k_{N}}\wedge y_{1}\wedge\ldots\wedge y_{k_{N}}$
for some $\{x_{j}\}_{j=1}^{k-k_{N}}\subset\mathcal{K}_{N-1}^{(0)}$ and $\{y_{j}\}_{j=1}^{k_{N}}\subset\mathcal{H}_{N}^{(0)}$.
This means that $\varphi_{N-1}^{-1}(v)=(x_{1}\wedge\ldots\wedge x_{k-k_{N}})\otimes(y_{1}\wedge\ldots\wedge y_{k_{N}})$.
The first factor is in $\ran\varphi_{N-2}$ and hence we can apply the lemma
to it, and so on, finally obtaining an $N$-fold tensor product of decomposable
elements in $\mathcal{H}_{i}$.
\end{proof}

This theorem applied to our scenario means that a state $\psi$ in $\mathcal{H}$ is
separable (tensor product of decomposable states) iff $\varphi(\psi)\in\mathcal{K}$
is separable (decomposable).

\subsection{Pl\"ucker relations}

Since a state in $\mathcal{K}$ is separable precisely when the Pl\"ucker
relations hold\cite{Hodge,Kasman,LPVP}, we can conclude that the Pl\"ucker relations provide a
sufficient and necessary condition of separability of an arbitrary system
of finitely many particles.

Recall that the Pl\"ucker relations say that 
a $k$ fermionic state with totally antisymmetric amplitudes $P_{j_1j_2\dots j_k}$ is separable iff for any $\mathcal{A}=\{a_{1},\ldots
,a_{k-1}\}\subset I$ and
$\mathcal{B}=\{b_{1},\ldots,b_{k+1}\}\subset I$ the polynomial
\begin{equation}
\Pi_{\mathcal{A},\mathcal{B}}(P)=\sum_{j=1}^{k+1}(-1)^{j-1}P_{a_{1}\ldots a_{k-1}b_{j}}P_{b_{1}\ldots b_{j-1}b_{j+1}\ldots b_{k+1}}
\end{equation}
equals to zero.
However, now our fermionic state is of special kind namely it is of the form
\begin{equation}
\varphi(\psi)=\sum_{J\subset I}P_{j_{1}j_{2}\ldots j_{k}}e_{j_{1}}\wedge e_{j_{2}}\wedge\ldots\wedge e_{j_{k}}
\end{equation}
\noindent
where $J=\{j_{1},\ldots,j_{k}\}$ and $I=\bigcup_{i=1}^{N}I_{i}=\{1,\ldots,n\}$ with
$I_{1}=\{1,\ldots,n_{1}\}, I_{2}=\{n_{1}+1,\ldots,n_{1}+n_{2}\}, \ldots, I_{
N}=\{n-n_{N}+1,\ldots,n\}$
($n=\sum_{i=1}^{N}n_{i}$ and $\{e_{j}\}_{j\in I_{i}}$ is an orthonormal
basis of $\mathcal{H}_{i}^{(0)}$).
Hence we do not need to consider all relations since many of them are
identically zero due to the very special form of the vectors in $\ran\varphi$.
This means that  $P_{j_{1}j_{2}\ldots j_{k}}=0$
if there exists $i\in\{1,\ldots,N\}$ such that $|J\cap I_{i}|\neq k_{i}$.
Using this property of the coefficients $P_{j_{1}\ldots j_{k}}$
one can see that $\Pi_{\mathcal{A},\mathcal{B}}(P)$ is identically zero unless
\begin{equation}
(\exists j\in\{1,\ldots,k+1\})(\forall i\in\{1,\ldots,N\})(|(\mathcal{A}\cup\{b_{j}\})\cap I_{i}|=|(\mathcal{B}\setminus\{b_{j}\})\cap I_{i}|=k_{i})
\end{equation}
which is equivalent to
\begin{multline}
(\exists j\in\{1,\ldots,k+1\})(\forall i\in\{1,\ldots,N\})( (b_{j}\notin I_{i})\textrm{ and }(|\mathcal{A}\cap I_{i}|=|\mathcal{B}\cap I_{i}|=k_{i}))  \\
\textrm{ or }((b_{j}\in I_{i})\textrm{ and }(|\mathcal{A}\cap I_{i}|+1=|\mathcal{B}\cap I_{i}|-1=k_{i})))
\end{multline}
When this holds for some $j$ we have exactly $|(\mathcal{B}\setminus\mathcal{A})\cap I_{i}|$
nonvanishing terms where $i$ is the unique index for which $b_{j}\in I_{i}$.

As a special case take the system of $N$ qubits described by the Hilbert
space $\mathcal{H}=\mathbb{C}^{2}\otimes\cdots\otimes\mathbb{C}^2$. With
the notations above this corresponds to $k_{1}=k_{2}=\ldots=k_{N}=1$ and
$n_{1}=n_{2}=\ldots=n_{N}=2$ therefore we can embed $\mathcal{H}$ in
$\mathcal{K}=\bigwedge^{N}\mathbb{C}^{2N}$. It is not too hard to check that
in this case the Pl\"ucker relations tell us that an element
\begin{equation}
\psi=\sum_{i_{1},\ldots,i_{N}=0}^{1}\psi_{i_{1}\ldots i_{N}} e_{i_{1}}\otimes\cdots\otimes e_{i_{N}}
\end{equation}
is separable if and only if for all $1\le j\le N$ and for all $(i_{1}, \ldots, i_{j-1}, i_{j+1}, \ldots, i_{N}, h_{1}, \ldots, h_{j-1},$ $h_{j+1}, \ldots, h_{N})\in\mathbb{Z}_{2}^{2N-2}$
we have 
\begin{equation}
\psi_{i_{1}\ldots i_{j-1}0 i_{j+1}\ldots i_{N}}\psi_{h_{1}\ldots h_{j-1}1 h_{j+1}\ldots h_{N}} = \psi_{i_{1}\ldots i_{j-1}1 i_{j+1}\ldots i_{N}}\psi_{h_{1}\ldots h_{j-1}0 h_{j+1}\ldots h_{N}}
\end{equation}
which is indeed the separability condition for $N$ qubits.
Notice that for $N=2$ (two qubits) we obtain the result of Gittings and Fischer\cite{Gittings}
of relating the concurrence to the fermionic measure of Schliemann\cite{Schlie} 
for a {\it special} two-fermion state with merely $4$ nontrivial amplitudes.

\subsection{Entangled states}

This scheme of embedding an arbitrary system of finitely many particles has
the property that separability in the larger space implies separability in
the smaller one. Since separable states are SLOCC equivalent to each other
we conclude that this special SLOCC equivalence class does not split into
subclasses when we restrict ourselves to the smaller space and its smaller
SLOCC group. Note however that this by no means is the case with other SLOCC
classes. For example look at the system of four qubits that can be embedded
in $\mathcal{K}=\bigwedge^{4}\mathbb{C}^{8}$. Let $\{e_{2i-1},e_{2i}\}$ be the
computational basis in the Hilbert space of the $i^{\mathrm{th}}$ qubit. Then the states
\begin{eqnarray}
P & = & \frac{1}{\sqrt{2}}e_{1}\otimes e_{3}\otimes(e_{5}\otimes e_{7}+e_{6}\otimes e_{8})  \\
Q & = & \frac{1}{\sqrt{2}}(e_{1}\otimes e_{3}+e_{2}\otimes e_{4})\otimes e_{5}\otimes e_{7}
\end{eqnarray}
cannot be transformed into each other since $P$ is $AB(CD)$-separable but
$Q$ is $(AB)CD$-separable. But their images in $\mathcal{K}$ are
 (for the definition of $\varphi$ see Eq. (23)).
\begin{eqnarray}
\varphi(P) & = & \frac{1}{\sqrt{2}}(e_{1}\wedge e_{3}\wedge e_{5}\wedge e_{7}+e_{1}\wedge e_{3}\wedge e_{6}\wedge e_{8})  \\
\varphi(Q) & = & \frac{1}{\sqrt{2}}(e_{1}\wedge e_{3}\wedge e_{5}\wedge e_{7}+e_{2}\wedge e_{4}\wedge e_{5}\wedge e_{7})
\end{eqnarray}
which can be transformed to each other by applying the element $X\otimes I\otimes I\in GL(8,\mathbb{C})$
therefore they are SLOCC-equivalent in $\mathcal{K}$.
(Here $X=\sigma_1$, i.e. the first Pauli matrix.)

This example convices us that the phenomenon of splitting of SLOCC classes is
not uncommon. In fact it does appear even in the three-qubit case where
$(AB)C$-, $(BC)A$- and $(CA)B$-biseparable states are inequivalent but their
images in the three-fermion system belong to the same class. However, this
embedding is rather special because there is "no room" for this kind of
splitting of multipartite entangled states and this may be a main reason of
the similarity of the classification of entanglement in the two systems.

Despite the fact that the entanglement measures in the embedding system may
be much coarser than needed for full SLOCC classification of the smaller
system, this method might prove to be a useful tool. If the splitting of SLOCC
classes could be fully understood than it would be enough to identify the
entanglement classes of $\bigwedge^{k}\mathbb{C}$ which might have a simpler
structure than the Hilbert space of a general system containing various types
of particles.

Let us now introduce a family of SLOCC-invariants for fermionic systems.
Let $\mathcal{K}^{(0)}=\mathbb{C}^{dk}$ and $\mathcal{K}=\bigwedge^{k}\mathcal{K}^{(0)}$
where $k\in 2\mathbb{N}$ and $d\in\mathbb{N}$. Given a state $P\in\mathcal{K}$
we can take the $d$-fold wedge product of it which lives in $\bigwedge^{dk}\mathbb{C}^{dk}$.
Hence it is invariant under the action of $SL(dk,\mathbb{C})$, and picks up a factor
corresponding to the determinant under the action of $GL(dk,\mathbb{C})$. Let $\xi(P)$
denote the absolute value of this vector:
\begin{equation}
\|P\wedge P\wedge\ldots\wedge P\|=\xi(P)
\end{equation}
Let $P_{i_{1}\ldots i_{k}}$ denote the coefficients of $P$ with respect to the
induced basis:
\begin{equation}
P = \sum_{i_{1}\ldots i_{k}=1}^{dk}P_{i_{1}\ldots i_{k}}e_{i_{1}}\wedge e_{i_{2}}\wedge\ldots\wedge e_{i_{k}}
\end{equation}
After expanding the product above we obtain
\begin{eqnarray}
P\wedge P\wedge\ldots\wedge P
  & = &  \sum_{i_{1},\ldots,i_{dk}=1}^{dk}P_{i_{1}\ldots i_{k}}P_{i_{k+1}\ldots i_{2k}}\ldots P_{i_{(d-1)k+1}\ldots i_{dk}}e_{i_{1}}\wedge e_{i_{2}}\wedge\ldots\wedge e_{i_{dk}}  \nonumber \\
  & = &  \sum_{\pi\in S_{dk}}\sigma(\pi)P_{\pi(1)\ldots\pi(k)}\ldots P_{\pi(dk-k+1)\ldots\pi(dk)}e_{1}\wedge e_{2}\wedge\ldots\wedge e_{dk}
\end{eqnarray}
where $S_{n}$ is the group of bijections form the set $\{1,\ldots,n\}$ to
itself, and $\sigma:S_{n}\to\{1,-1\}$ is the alternating representation of
this group. From this one can see that
\begin{equation}
\xi(P)=\left|\sum_{\pi\in S_{dk}}\sigma(\pi)P_{\pi(1)\ldots\pi(k)}P_{\pi(k+1)\ldots\pi(2k)}\ldots P_{\pi(dk-k+1)\ldots\pi(dk)}\right|
\end{equation}
Since $\xi((A\otimes\cdots\otimes A)P)=|\det A|\xi(P)$,
it follows that if $P'=(A\otimes\cdots\otimes A)P$ then either both of $\xi(P)$
and $\xi(P')$ are $0$ or none of them. Moreover, if $|\det A|=1$ then $\xi(P)=\xi(P')$.

Now for $N\in 2\mathbb{N}$ and $d\in\mathbb{N}$ we can embed the Hilbert space
$\mathcal{H}=\mathbb{C}^{d}\otimes\cdots\otimes\mathbb{C}^{d}$ of $N$ qudits
into $\mathcal{K}$, the image of
\begin{equation}
\psi = \sum_{i_{1},\ldots,i_{N}=1}^{d}\psi_{i_{1}i_{2}\ldots i_{N}}e_{i_{1}}\otimes\ldots\otimes e_{i_{N}}
\end{equation}
being
\begin{equation}
\tilde{\varphi}(\psi)=P = \sum_{i_{1},\ldots,i_{N}=1}^{d}\psi_{i_{1}i_{2}\ldots i_{N}} e_{i_{1}}\wedge e_{d+i_{2}}\wedge\ldots\wedge e_{(N-1)d+i_{N}}
\end{equation}
Note that the definition of $\tilde{\varphi}$ is slightly different from that
of the map $\varphi$ defined above, but the difference is only a relabelling of
basis elements in $\mathcal{H}$. For this state the value of $\xi(P)$ is
\begin{equation}
\xi(P)=\left|\sum_{\pi_{1},\ldots,\pi_{N}\in S_{d}}\left(\prod_{i=1}^{N}\sigma(\pi_{i})\right)\psi_{\pi_{1}(1)\pi_{2}(1)\ldots\pi_{N}(1)}\psi_{\pi_{1}(2)\pi_{2}(2)\ldots\pi_{N}(2)}\ldots\psi_{\pi_{1}(d)\pi_{2}(d)\ldots\pi_{N}(d)}\right|
\end{equation}
With a slight abuse of notation $\xi(\psi)$ will denote $\xi(\tilde{\varphi}(\psi))$.
Note that for the special case of an {\it even} number of qubits i.e. $d=2$
the squared magnitude of the measure of Eq. (42) is related to the one of Wong and Christensen\cite{Wong}.
For four qubits it is known that this measure boils down to the one denoted by the letter $H$ in the paper of Luque and Thibon\cite{Luque}.

As an example take the following two globally entangled states in $\mathcal{H}$:
\begin{eqnarray}
\psi & = & \frac{1}{\sqrt{d}}(e_{1}\otimes\cdots\otimes e_{1}+e_{2}\otimes\cdots\otimes e_{2}+\ldots+e_{d}\otimes\cdots\otimes e_{d})  \\
\phi & = & \frac{1}{\sqrt{N}}(e_{2}\otimes e_{1}\otimes\cdots\otimes e_{1}+e_{1}\otimes e_{2}\otimes e_{1}\otimes\cdots\otimes e_{1}+\ldots+e_{1}\otimes\cdots\otimes e_{1}\otimes e_{2})
\end{eqnarray}
Then we have
\begin{equation}
\xi(\psi) = \sum_{\pi\in S_{d}}\left(\frac{1}{\sqrt{d}}\right)^{d}=d!d^{-\frac{d}{2}}\neq 0
\end{equation}
but $\xi(\phi) = 0$ for $N>2$ hence we can conclude that $\psi$ and $\phi$ are not SLOCC equivalent.

\subsection{Reduced density matrices}

In many cases one can gain information about a fermionic system by looking
at its single particle reduced density matrix. The mapping described above
takes a pure state of an arbitrary system and maps it to a special fermionic
one therefore the question naturally arises: how are the one particle reduced
density matrices of the two states related to each other?

Let $\mathcal{H}_{i}=\bigwedge^{k_{i}}\mathcal{H}_{i}^{(0)}$,
$\mathcal{H}=\mathcal{H}_{1}\otimes\cdots\otimes\mathcal{H}_{N}$,
$\mathcal{K}^{(0)}=\mathcal{H}_{1}^{(0)}\oplus\cdots\oplus\mathcal{H}_{N}^{(0)}$
and $\mathcal{K}=\bigwedge^{k}\mathcal{K}^{(0)}$ as before ($k=k_{1}+\ldots+k_{N}$), and let
$\{e_{a}\}_{a\in I_{i}}$ be orthonormal bases in $\mathcal{H}_{i}^{(0)}$
respectively where $\dim\mathcal{H}_{i}^{(0)}=|I_{i}|=n_{i}$ ($i\in\{1,\ldots,N\}$),
$n:=n_{1}+\cdots+n_{N}$. The function
$\varphi:\mathcal{H}\to\mathcal{K};(v_{1}\wedge\ldots\wedge v_{k_{1}})\otimes\cdots\otimes(v_{k-k_{N}+1}\wedge\ldots\wedge v_{k})\mapsto v_{1}\wedge\ldots\wedge v_{k}$
maps a general state
\begin{equation}
\psi=\sum_{\substack{i_{1},\ldots,\\ i_{k_{1}}\in I_{1}}}\sum_{\substack{i_{k_{1}+1},\ldots, \\i_{k_{1}+k_{2}}\in I_{2}}}\cdots\sum_{\substack{i_{k-k_{N}+1},\\ \ldots,i_{k}\in I_{N}}}\psi_{i_{1}i_{2}\ldots i_{k}}(e_{i_{1}}\wedge\ldots\wedge e_{i_{k_{1}}})\otimes\cdots\otimes(e_{i_{k-k_{N}+1}}\wedge\ldots\wedge e_{i_{k}})
\end{equation}
to
\begin{equation}
\varphi(\psi)=P=\sum_{\substack{i_{1},\ldots,\\ i_{k_{1}}\in I_{1}}}\sum_{\substack{i_{k_{1}+1},\ldots, \\i_{k_{1}+k_{2}}\in I_{2}}}\cdots\sum_{\substack{i_{k-k_{N}+1},\\ \ldots,i_{k}\in I_{N}}}\psi_{i_{1}i_{2}\ldots i_{k}}e_{i_{1}}\wedge\ldots\wedge e_{i_{k}}
\end{equation}

Let $\rho$ denote the one particle density matrix of the state $P$ and $\rho_{i}$
$(i\in\{1,\ldots,N\})$:
\begin{eqnarray}
\rho & = & \Tr_{2,3,\ldots,k}PP^{*}=(id_{\mathcal{K}^{(0)}}\otimes\Tr\otimes\cdots\otimes\Tr)PP^{*}\in Mat(n,\mathbb{C})  \nonumber \\
\rho_{1} & = & \Tr_{2,3,\ldots,k}\psi\psi^{*}=(id_{\mathcal{H}_{1}^{(0)}}\otimes\Tr\otimes\cdots\otimes\Tr)\psi\psi^{*}\in Mat(n_{1},\mathbb{C})  \nonumber \\
  & \vdots &  \nonumber \\
\rho_{N} & = & \Tr_{1,2,\ldots,k-k_{N},k-k_{N}+2,\ldots,k}\psi\psi^{*}\in Mat(n_{N},\mathbb{C})
\end{eqnarray}

Using identities like
\begin{eqnarray}
\Tr_{2,3,\ldots,k}(e_{i_{1}}\wedge\ldots\wedge e_{i_{k}})(e_{i_{1}}^{*}\wedge\ldots\wedge e_{i_{k}}^{*}) & = & \frac{1}{k}(e_{i_{1}}e_{i_{1}}^{*}+\ldots+e_{i_{k}}e_{i_{k}}^{*})  \nonumber \\
\Tr_{2,3,\ldots,k}(e_{i_{1}}\wedge\ldots\wedge e_{i_{k}})(e_{i'_{1}}^{*}\wedge\ldots\wedge e_{i_{k}}^{*}) & = & \frac{1}{k}e_{i_{1}}e_{i'_{1}}^{*}  \nonumber \\
\Tr_{2,3,\ldots,k}[(e_{i_{1}}\wedge\ldots\wedge e_{k_{1}})\otimes\cdots\otimes(e_{i_{k-k_{N}+1}}\wedge\ldots\wedge e_{i_{k}})] {} & & \nonumber \\
 {} [(e_{i_{1}}^{*}\wedge\ldots\wedge e_{k_{1}}^{*})\otimes\cdots\otimes(e_{i_{k-k_{N}+1}}^{*}\wedge\ldots\wedge e_{i_{k}}^{*})] & = & \frac{1}{k_{1}}(e_{i_{1}}e_{i_{1}}^{*}+\ldots+e_{i_{k_{1}}}e_{i_{k_{1}}}^{*})  \nonumber \\
\Tr_{2,3,\ldots,k}[(e_{i_{1}}\wedge\ldots\wedge e_{k_{1}})\otimes\cdots\otimes(e_{i_{k-k_{N}+1}}\wedge\ldots\wedge e_{i_{k}})] {} & & \nonumber \\
 {} [(e_{i'_{1}}^{*}\wedge\ldots\wedge e_{k_{1}}^{*})\otimes\cdots\otimes(e_{i_{k-k_{N}+1}}^{*}\wedge\ldots\wedge e_{i_{k}}^{*})] & = & \frac{1}{k_{1}}e_{i_{1}}e_{i'_{1}}^{*}
\end{eqnarray}
we can see that
\begin{equation}
\rho=\bigoplus_{i=1}^{N}\frac{k_{i}}{k}\rho_{i}
\end{equation}
Since both sides are linear in the density matrices of the whole systems, this
holds for mixed states too. An alternative normalization is given by $\gamma=k\rho$
and $\gamma_{i}=k_{i}\rho_{i}$, the relation for these is
$\gamma=\bigoplus_{i=1}^{N}\gamma_{i}$ ($i\in\{1,\ldots,N\}$). Observe that the
two states are unentangled iff $\gamma_{i}^{2}=\gamma_{i}$ ($i\in\{1,\ldots,N\}$)
 and $\gamma^{2}=\gamma$ respectively, yielding an alternative proof of our theorem.

\subsection{Physical interpretation}

So far we have regarded this relation of systems of distinguishable particles
and fermionic ones as a purely mathematical one. However, we can interpret it
as a physical property of the particles pretending they are \emph{really}
indistinguishable but for some reason they are not in the same state of some
inner degree of freedom (analogous to isospin) and these inner states are not
mixed by the Hamiltonian.

As the simplest example suppose we have two distinguishable qubits so each
particle has two states, call them $e_{0},e_{1}$ and
$e'_{0},e'_{1}$. We can combine the four states of the two particles
into a single Hilbert space having dimension $4$. If the Hamiltonian governing
the evolution of a state in this space has vanishing matrix elements between
basis states with and without a prime then the subspaces spanned by
$\{e_{0},e_{1}\}$ and $\{e'_{0},e'_{1}\}$ are not mixed and
therefore we can associate an inner quantum number to the states. A state
which is a linear combination of $e_{0}$ and $e_{1}$ can be called a
particle of type one and a state in the span of $\{e'_{0},e'_{1}\}$
can be called a particle of type two. If we take two fermionic particles having
these four single particle states and one of them is of type one and the other
is of type two then they will behave exactly as if they were distinguishable
qubits.

\section{Conclusions}

In this paper we have studied quantum systems containing both distinguishable and identical constituents.
A special subclass of such systems can be studied using the algebraic constructs called Freudenthal systems.
The corresponding physical systems of this subclass are the {\it tripartite } ones that can be embedded to a system consisting of three fermions with six single particle states. Such embedded systems are the ones consisting of a qubit and a bipartite fermionic system with four single particle states, three ordinary qubits,  three bosonic qubits, and two bosonic qubits coupled to an ordinary qubit.
For these systems we presented a {\it complete} classification of SLOCC orbits,
based on the quartic SLOCC invariant arising from the corresponding one of the associated Freudenthal system.
Though the full appreciation of these invariants within the field of quantum information is still missing, we have pointed out that they arise quite naturally
as entropy formulas for black hole solutions in supergravity theories.

As a next step retaining merely the idea of embedding one type of system to the other we studied issues of separability for systems embedded into fermionic ones. We proved that the Pl\"ucker relation for these embedding fermionic systems play an universal role in checking the separability of the embedded ones.
We briefly elaborated also on the interesting problem of splitting of SLOCC classes
when comparing the entanglement properties of the embedding and embedded systems. Such considerations enabled a construction of a class of pure state entanglement measures containing some well-known ones as a limiting case.
Since in many cases we can gain information about a fermionic system by looking at its single particle reduced density matrices, a natural question to be addressed is the one: how these density matrices for the embedding and embedded systems are related?
We have answered this question by presenting an explicit formula.
And at last we proposed a possible physical interpretation of our embedding of one type of system to the other.
The conclusion is that we can regard our embedding trick
as a convenient representation for a quantum system with particles which are {\it really} indistinguishable but for some reason they are not in the same state of some {\it inner} degree of freedom.

\section{Appendix: Cubic Jordan algebras and Freudenthal triples}

\begin{defn}
An algebra (not necessarily associative) $(J,+,\bullet)$ is called a Jordan
algebra if it is commutative and for any two elements $A,B\in J$ the equation
$(A\bullet A)\bullet(A\bullet B)=A\bullet((A\bullet A)\bullet B)$ holds.

A Jordan algebra is cubic if every element satisfies a cubic polynomial equation.
\end{defn}

The Springer construction of cubic Jordan algebras tells us that one can obtain
a cubic Jordan algebra starting with a vector space $V$ equipped with a
suitable cubic form $N:V\to\mathbb{C}$ and a basepoint $c\in V$ such that
$N(c)=1$. One can then define various maps using the linearization
\begin{equation}
N(x,y,z)=\frac{1}{6}\big(N(x+y+z)-N(x+y)-N(x+z)-N(y+z)+N(x)+N(y)+N(z)\big)
\end{equation}
of N, including the Jordan product, but for our purposes only the following two
are needed:
\begin{eqnarray}
(\cdot,\cdot):V\times V\to\mathbb{C} &;& (x,y)=9N(c,c,x)N(c,c,y)-6N(x,y,c)  \\
\cdot^{\sharp}:V\to V &;& \forall y\in J: (x^{\sharp},y)=3N(x,x,y)
\end{eqnarray}
The former is called the trace bilinear form, while the latter is the adjoint
or sharp map.

From a Jordan algebra $J$ over $\mathbb{C}$ one can obtain the Freudenthal
triple system $\mathfrak{M}(J)=\mathbb{C}\oplus\mathbb{C}\oplus J\oplus J$
which is equipped with a skew-symmetric bilinear form and a quartic form
defined by:
\begin{eqnarray}
\{x,y\} & = & \alpha\delta-\beta\gamma+(A,D)-(B,C)  \\
q(x) & = & 2\big((A,B)-\alpha\beta\big)^{2}-8(A^{\sharp},B^{\sharp})+8\alpha N(A)+8\beta N(B)
\end{eqnarray}
where $x=(\alpha,\beta,A,B)$ and $y=(\gamma,\delta,C,D)$ are two elements of
$\mathfrak{M}(J)$. One can also define the unique trilinear map
$T:\mathfrak{M}(J)\times \mathfrak{M}(J)\times \mathfrak{M}(J)\to \mathfrak{M}(J)$
with the property $\{T(x,y,z),w\}=q(x,y,z,w)$ where $q(\cdot,\cdot,\cdot,\cdot)$ is the linearization
of the quartic form $q(\cdot)$.

\begin{defn}
$\Inv(\mathfrak{M}(J))$ is the group of linear transformations which preserve
these forms, i.e. for all $\sigma\in\Inv(\mathfrak{M}(J))$
\begin{equation}
\{\sigma(\cdot),\sigma(\cdot)\}=\{\cdot,\cdot\}\quad\textrm{and}\quad q\circ\sigma=q
\end{equation}
holds.
\end{defn}

Clearly, the construction yields a $2+2\dim J$ dimensional representation of
$\Inv(\mathfrak{M}(J))$ and $q$ is a quartic polynomial invariant under the
action of this group. In the following, we give explicitely the Jordan algebras
needed for the classification of entangled states in the above mentioned quantum
systems. It turns out that all of them can be regarded as a subalgebra of
$J_{3}=M(3,\mathbb{C})$ so we start with this one.

For $A\in J_{3}$ the cubic form $N$ is simply the determinant of the $3\times 3$
matrix, for $A,B\in J_{3}$ the trace bilinear form is given by $(A,B)=\Tr(AB)$
and the explicit form of the sharp map is
\begin{equation}
A^{\sharp}=A^{2}-\Tr(A)A+\frac{1}{2}\big(\Tr(A)^{2}-\Tr(A^{2})\big)I_{3}
\end{equation}

The simplest nontrivial Jordan algebra is $J_{1}=\mathbb{C}$, the cubic
norm of $a\in J_{1}$ is $N(a)=a^{3}$. It follows that the map
$J_{1}\to J_{3};a\mapsto aI_{3}$ is an injective morphism of cubic Jordan
algebras. The next Jordan algebra is $J_{1+1}=\mathbb{C}\oplus\mathbb{C}$,
the value of $N$ on the element $x=(a,b)$ is $N(x)=ab^{2}$. In this case
the image of $x$ in $J_{3}$ is
\begin{equation}
\left[\begin{array}{ccc}
a & 0 & 0  \\
0 & b & 0  \\
0 & 0 & b
\end{array}\right]
\end{equation}
The third Jordan algebra is $J_{1+1+1}=\mathbb{C}\oplus\mathbb{C}\oplus\mathbb{C}$
Here for $x=(a,b,c)$ the value of $N$ is $N(x)=abc$. This is nothing else but
the determinant of the matrix
\begin{equation}
\left[\begin{array}{ccc}
a & 0 & 0  \\
0 & b & 0  \\
0 & 0 & c
\end{array}\right]
\end{equation}
which shows us the isomorphism between $J_{1+1+1}$ and the subalgebra of diagonal
matrices in $J_{3}$. The last Jordan algebra we consider is
$J_{1+2}=\mathbb{C}\oplus Q_{4}$, where $Q_{4}$ is any $4$ dimensional complex
vector space with a nondegenerate quadratic form. It is convinient to let $Q_{4}$
be the vector space of $2\times 2$ matrices, and the quadratic form be the
determinant. A general element in $J_{1+2}$ is therefore $x=(a,A)$, and its cubic
norm is $N(x)=a\det A$. For this Jordan algebra the inclusion map is given by
\begin{equation}
(a,
\left[\begin{array}{cc}
A_{11} & A_{12}  \\
A_{21} & A_{22}
\end{array}\right])
\mapsto
\left[\begin{array}{ccc}
a & 0 & 0  \\
0 & A_{11} & A_{12}  \\
0 & A_{21} & A_{22}
\end{array}\right]
\end{equation}
the image being a block diagonal matrix built from a $1\times 1$ and a $2\times 2$
block.

What makes these constructions useful for studying entanglement is the fact
that the $\Inv$ groups are almost SLOCC groups of various quantum systems.
Namely, $\Inv(\mathfrak{M}(J_{1}))\simeq SL(2,\mathbb{C})$,
$\Inv(\mathfrak{M}(J_{1+1}))\simeq SL(2,\mathbb{C})^{2}$,
$\Inv(\mathfrak{M}(J_{1+1+1}))\simeq SL(2,\mathbb{C})^{3}$,
$\Inv(\mathfrak{M}(J_{1+2}))$ is isomorphic to $SL(2,\mathbb{C})\times SL(4,\mathbb{C})$
and finally
$\Inv(\mathfrak{M}(J_{3}))\simeq SL(6,\mathbb{C})$. The SLOCC groups are obtained
by replacing $SL(n,\mathbb{C})$ with $GL(n,\mathbb{C})$. Excluding the $0$ vector
we can identify four SLOCC orbits (characterized by the rank of vectors) in each
case except the first one where the rank 2 orbit is absent. Vectors of rank 4 are
the ones for which $q$ does not vanish, all others are of rank at most 3. $x$ has
rank 3 iff $q(x)=0$ and $T(x,x,x)\neq 0$. Vectors with rank 2 are the ones for
which $T(x,x,x)$ vanishes but there exists $y$ such that $3T(x,x,y)+\{x,y\}x\neq 0$.
If there is no such $y$ then $x$ has rank 1.

                                                                                \section{Acknowledgements}                                                      Financial support from the Orsz\'agos Tudom\'anyos Kutat\'asi Alap              (grant numbers T047035, T047041, T038191) is                                    gratefully acknowledged.                                                        

\end{document}